\documentclass[11pt]{article}

\usepackage[margin=1in]{geometry}
\usepackage{amsmath,amssymb,amsthm}
\usepackage{hyperref}
\usepackage{enumitem}
\usepackage{setspace}

 \usepackage{graphicx}
\usepackage{natbib}

\newtheorem{proposition}{Proposition}
\newtheorem{lemma}{Lemma}

\theoremstyle{definition}
\newtheorem{assumption}{Assumption}

\title{Data Sharing and Competition in Learning-by-Deploying Industries\\
{\large\textit{Insights from Robotics and Beyond}}}
\author{Yunjin Tong\\ \normalsize Stanford Graduate School of Business
\and Luca-Andrei Manea\\ \normalsize Stanford Graduate School of Business}
\date{}

\begin{document}
\maketitle

\section{Introduction}

Many modern technologies improve through use. Each unit deployed in the field produces output and, at the same time, generates data that is fed back to train and refine the technology itself, so that the next generation of units is more productive. Autonomous vehicles, logistics systems, predictive maintenance, industrial inspection, and robots all share this feature: deployment is not only production but also an investment in a shared, cumulative learning stock. When that data is pooled across many deployments, the productivity gains spill over to the entire fleet. When it is siloed, each operator learns only from its own history. The architecture of learning, pooled versus fragmented, therefore shapes how fast the technology improves and how the gains are distributed. \smallskip \\
Robots are the most current example. Companies such as Agility Robotics, Figure, Unitree, and Tesla are deploying physical units in warehouses and factories, while the control policies that run them are increasingly trained as foundation models that pool deployment data across an entire fleet. A unit deployed today produces output and generates training data that increases the productivity of every unit tomorrow. This learning structure now sits at the center of an industrial-policy contest. China's 15th Five-Year Plan (2026--2030) elevates ``embodied intelligence'' to a top-line national priority, backed by coordinated procurement, dedicated standardization committees, and vertically integrated platforms that consolidate deployment data across many sites. By 2024 China already installed over half of the world's industrial robots and operated a stock exceeding two million units \citep{thediplomat_2026,uscc_2024}. The United States, by contrast, leads in AI software and high-end research but deploys through a more fragmented landscape of competing platforms and decentralized purchasers, with a comparatively piecemeal policy response \citep{uscc_2024}. The conventional framing of this contrast is static, a race over unit costs and supply chains. But the underlying object is dynamic. What matters is how the architecture of learning interacts with firms' adoption decisions over time and which policy levers shift the resulting equilibrium. \smallskip \\
We study this with a deliberately simple two-period model. 
Symmetric firms make irreversible capacity decisions. Capacity 
in use generates data that feeds a learning curve, which increases 
next-period productivity. The data may be pooled, with all firms 
drawing on a common stock generated by joint deployment, or 
fragmented, with each firm learning only from its own history. 
We follow the learning-by-doing tradition of \citet{arrow_1962}, 
but replace production experience with deployment data as the 
source of productivity growth, a distinction we term 
\emph{learning-by-deploying}. While robots are one motivating 
example, the model applies to any learning-by-deploying industry 
in which use generates feedback data that improves a shared 
technology.
\paragraph{Contribution.} We isolate when pooled learning is socially valuable, when firms will sustain it voluntarily, and how product-market competition changes both answers. In a baseline with an exogenous output price, pooling is unambiguously beneficial and firms underinvest in early deployment relative to a social planner. Once the price is endogenized through downstream Cournot competition, this clean prescription breaks down. Pooling raises every firm's output at once, which decreases the price, so the private value of sharing falls with competition intensity and can turn negative. We characterize a sharing-sustainability threshold, show it is governed by the elasticity of industry demand over the output range that pooling induces, and use a fully numerical solution under general demand to confirm that the qualitative patterns survive once the tractability assumptions are dropped.

\paragraph{Policy implications.} The analysis yields the following points. \\
\textit{Without product-market competition.}
    \begin{itemize}[itemsep=0pt, topsep=0.2em]
    \item Pooling raises welfare but firms underinvest relative to 
    the social optimum, so policy should encourage sharing and 
    subsidize adoption.
    \item Whether pooling raises or lowers early deployment depends 
    on the persistence of the learning curve, not its steepness, 
    so a uniform early-adoption subsidy may be misdirected.
    \end{itemize}
\textit{With product-market competition.}
    \begin{itemize}[itemsep=0pt, topsep=0.2em]
    \item Pooling raises aggregate output and depresses price, so 
    the private gain from sharing falls with competition intensity 
    and can turn negative.
    \item Two industry primitives shape the value of sharing. The 
    first is the persistence of the learning curve. The competitive 
    cost of sharing is large and durable where learning remains 
    steep past the fragmented deployment level, and small where it 
    saturates early.
    \item The second is demand elasticity. The pooling gain depends 
    on elasticity over the relevant output range and need not fall 
    monotonically with the number of competitors. Blanket mandates 
    such as data interoperability are therefore misplaced.
    \end{itemize}
This paper is organized as follows. Section~\ref{sec:simple} develops the 
baseline two-period model with exogenous prices and characterizes 
the welfare ladder and underinvestment results. 
Section~\ref{sec:endo-sharing} endogenizes both the output price 
and the data-sharing decision, derives the sustainability threshold, 
and shows how product-market competition qualifies the baseline 
prescriptions. Section~\ref{sec:numerical-general} solves the full 
model numerically under general inverse demand and confirms that the 
qualitative patterns survive. All proofs are collected in Appendix~\ref{app:proofs}.

\section{Related Work}
This paper contributes to work on production settings with learning curves, where the productivity or cost of a technology changes with cumulative use. The classic learning-by-doing literature treats experience as a source of productivity growth \citep{arrow_1962}, and subsequent empirical and operations work documents both persistence and transfer of learning across industrial settings \citep{argote_epple_1990,argote_beckman_epple_1990,benkard_2000}. A related capacity-investment literature studies the timing, irreversibility, and flexibility of deployment decisions \citep{fine_freund_1990,vanmieghem_2003}. Our model links these two objects. Early capacity is not only productive capacity, but also an input into the future learning state. \smallskip \\
A second line of work studies strategic learning, spillovers, and cooperation under product-market competition. Dynamic models of learning-by-doing show that experience accumulation can reshape market conduct, pricing, and welfare \citep{spence_1981,fudenberg_tirole_1983,ghemawat_spence_1985,cabral_riordan_1994}. The R\&D-spillover and research-joint-venture literature studies when firms internalize or share knowledge generated by investment \citep{daspremont_jacquemin_1988,kamien_muller_zang_1992}. We differ by focusing on the topology of the learning system itself: firms either learn only from their own deployment histories or draw from a pooled stock generated by all adopters. The resulting externality is dynamic and operational, and the private value of pooling depends on the same product-market competition that pooling intensifies. \smallskip \\
We also relate to the economics of data. Data is non-rival and broader access can raise welfare even when firms privately prefer to hoard it \citep{jones_tonetti_2020}. Social and customer data can also create externalities and competitive advantages \citep{bergemann_bonatti_gan_2022,hagiu_wright_2023,pruefer_schottmueller_2021}. Much of this work studies prediction, digital products, or consumer-data markets. Here, the data are generated as a byproduct of productive deployment, so adoption, output, and future model quality are jointly determined. This makes the shape of the learning curve, especially how long marginal learning persists, central to both welfare and the sustainability of sharing. \smallskip \\
Robotics is one motivating example. Recent robot-learning systems and cross-embodiment datasets illustrate how trajectories from many tasks or platforms can train reusable control policies \citep{brohan_2022_rt1,open_x_embodiment_2023}. The same structure can arise in other settings with deployment-generated feedback data, including autonomous vehicles, logistics, predictive maintenance, and industrial inspection. The contribution is to isolate when pooled learning is welfare-improving and when product-market competition can instead make fragmentation privately stable and potentially socially relevant.

\section{Simple Exogenous-Price Model}
\label{sec:simple}
This section develops the simple baseline. Two symmetric firms each operate a stock of installed capacity and capacity in use generates data that improves a shared technology whose productivity rises with cumulative experience. The model is meant to describe any learning-by-doing setting in which deployed units produce data that feeds into a common stock of knowledge. Robotics is one instance but the formal content does not depend on the underlying technology. \smallskip \\
Two symmetric firms $i \in \{1, 2\}$ make irreversible capacity decisions in two periods $t \in \{1, 2\}$. The per-unit revenue is $r > 0$. Capacity acquisition has quadratic cost, but the cost coefficients may differ across periods. The total cost in period 1 is $(c_1/2)K_{i,1}^2$, while period 2 has incremental capacity cost $(c_2/2)(K_{i,2}-K_{i,1})^2$, with
\[
    K_{i,2} \geq K_{i,1}, \qquad c_1>0,\quad c_2>0.
\]
The discount factor is $\delta \in (0, 1]$.

\paragraph{Lagged pooled learning.} Each unit of capacity produces output at productivity $\eta(N_t)$, where $N_t$ is the cumulative learning state at the start of period $t$. The interpretation is that capacity in use generates deployment data, which is processed between periods to update the shared technology, so that all units operating next period benefit from its higher productivity. Specialized to the two-period setting:
\begin{itemize}[leftmargin=1.5em, nosep]
    \item Period 1: no prior data, so $N_1 = 0$ and productivity is the baseline $\eta(0) = \eta_{\min}$.
    \item Period 2: data from period-1 deployment has been processed and used to update the shared technology. The learning state $N_2$ depends on the data-pooling regime:
    \begin{itemize}[leftmargin=1.5em, nosep]
        \item \textbf{Fragmented:} $N_{2,i} = K_{i,1}$. Each firm benefits only from its own data.
        \item \textbf{Pooled:} $N_2 = K_{1,1} + K_{2,1}$. Joint deployment produces a pooled dataset on which both firms draw.
    \end{itemize}
\end{itemize}
Period-2 deployment data would affect period-3 productivity, but there is no period 3 in this model. Thus $K_{i,2}$ affects period-2 output but not period-2 learning. The learning curve $\eta : \mathbb{R}_{\geq 0} \to \mathbb{R}_{\geq 0}$ is twice continuously differentiable, strictly increasing, strictly concave, and bounded, except in the constant-learning benchmark considered below.

\paragraph{Profit.} Firm $i$'s discounted profit is
\[
    \pi_i =
    \underbrace{r K_{i,1}\eta_{\min}
    -\frac{c_1}{2}K_{i,1}^2}_{\text{Period-1 Payoff}}
    +\underbrace{\delta\left[
        r K_{i,2}\eta(N_2)
        -\frac{c_2}{2}(K_{i,2}-K_{i,1})^2
    \right]}_{\text{Period-2 Payoff}}.
\]
Period-1 capacity has two roles. It earns baseline period-1 revenue and it generates data that increases period-2 productivity. Under pooling, the data-generating role creates an externality on the rival's period-2 productivity.

\paragraph{Three regimes.} We compare three subgame-perfect equilibrium concepts at the symmetric solution $K_{1,t} = K_{2,t} = K_t$:
\begin{itemize}[leftmargin=1.5em, nosep]
    \item \textbf{(NF) Fragmented Nash.} Each firm's two-period problem decouples; the symmetric Nash equilibrium coincides with the symmetric fragmented social planner.
    \item \textbf{(NP) Pooled Nash.} Each firm takes the rival's period-1 capacity as given. Each firm benefits from pooled productivity $\eta(K_{1,1}+K_{2,1})$ in period 2, but ignores the benefit of its own period-1 data on the rival.
    \item \textbf{(SP) Pooled Social Planner.} A planner chooses symmetric capacities $(K_1,K_2)$ to maximize total welfare $W=\pi_1+\pi_2$ under the pooled learning rule.
\end{itemize}
Denote equilibrium capacities $(K^*_{\bullet,1},K^*_{\bullet,2})$ and welfares $W^*_{\bullet}$ for $\bullet \in \{\mathrm{NF},\mathrm{NP},\mathrm{SP}\}$.

\subsection{Reduced-Form Characterization}
\label{ssec:reduced-form}

The two-period problem collapses to a one-dimensional problem in period-1 capacity. We first solve the period-2 subproblem given period-1 capacity and the period-2 learning state.
\begin{lemma}[Period-2 capacity rule]\label{lem:period2}
Fix period-1 capacity $K_{i,1}$ and the period-2 learning state $N$ with $\eta(N)>0$. The period-2 problem
\[
    \max_{K_{i,2}\geq K_{i,1}}
    rK_{i,2}\eta(N)-\frac{c_2}{2}(K_{i,2}-K_{i,1})^2
\]
has the unique solution
\[
    K_{i,2}^*(K_{i,1},N)\;=\;K_{i,1}+\frac{r\eta(N)}{c_2}.
\]
\end{lemma}
\noindent Substituting Lemma~\ref{lem:period2} into the period-2 payoff defines the per-firm optimized continuation value
\begin{equation}\label{eq:phi-def}
    \phi(K, N)
    \;:=\;
    \underbrace{rK\,\eta(N)}_{\text{revenue from already-installed capacity}}
    \;+\;
    \underbrace{\frac{r^2}{2c_2}\,\eta(N)^2}_{\text{optimized value of period-2 expansion}}.
\end{equation}
The map $N \mapsto \phi(K,N)$ is non-decreasing for any $K\geq 0$ and strictly increasing whenever $\eta'(N)>0$, because
\[
    \partial_N \phi(K,N) \;=\; \eta'(N)\Bigl[\,rK + \tfrac{r^2}{c_2}\eta(N)\Bigr] \;\geq\; 0.
\]
This monotonicity is the only property of $\phi$ used in the welfare-ladder lemma below. \smallskip \\
Now, each regime is determined by how period-1 capacities determine the period-2 learning state $N$. At a symmetric profile $K_{1,1}=K_{2,1}=x$ (and, for pooled Nash, at an asymmetric profile $(x_i, x_j)$), we have
\begin{align}
    V_{\mathrm{NF}}(x)
    &\;:=\; r\eta_{\min}x - \tfrac{c_1}{2}x^2 + \delta\,\phi(x,\,x),
        \label{eq:V-NF}\\
    V_{\mathrm{SP}}(x)
    &\;:=\; r\eta_{\min}x - \tfrac{c_1}{2}x^2 + \delta\,\phi(x,\,2x),
        \label{eq:V-SP}\\
    V_i^{\mathrm{NP}}(x_i,x_j)
    &\;:=\; r\eta_{\min}x_i - \tfrac{c_1}{2}x_i^2 + \delta\,\phi(x_i,\,x_i+x_j).
        \label{eq:V-NP}
\end{align}
By construction, $V_i^{\mathrm{NP}}(x,x) = V_{\mathrm{SP}}(x)$ at any symmetric profile. At symmetric play the per-firm pooled-Nash payoff equals the per-firm planner value. \smallskip \\
Then, differentiating each reduced objective with respect to firm $i$'s period-1 capacity at the symmetric point yields one FOC per regime,
\begin{align}
    F_{\mathrm{NF}}(x)
    &=
    r\eta_{\min}-c_1 x
    +\delta\!\left[
        r\bigl(\eta(x)+x\eta'(x)\bigr)
        +\frac{r^2}{c_2}\eta(x)\eta'(x)
    \right], \label{eq:F-NF}\\
    F_{\mathrm{NP}}(x)
    &=
    r\eta_{\min}-c_1 x
    +\delta\!\left[
        r\bigl(\eta(2x)+x\eta'(2x)\bigr)
        +\frac{r^2}{c_2}\eta(2x)\eta'(2x)
    \right], \label{eq:F-NP}\\
    F_{\mathrm{SP}}(x)
    &=
    r\eta_{\min}-c_1 x
    +\delta\!\left[
        r\bigl(\eta(2x)+2x\eta'(2x)\bigr)
        +\frac{2r^2}{c_2}\eta(2x)\eta'(2x)
    \right]. \label{eq:F-SP}
\end{align}
Here $F_{\mathrm{NF}}=V_{\mathrm{NF}}'$ is the fragmented FOC (Nash and planner coincide), $F_{\mathrm{NP}}=\partial_{x_i} V_i^{\mathrm{NP}}\!\bigl|_{x_i=x_j=x}$ is the pooled-Nash FOC, and $F_{\mathrm{SP}}=V_{\mathrm{SP}}'$ is the pooled-planner FOC. $F_{\mathrm{NP}}$ differs from $F_{\mathrm{NF}}$ only in evaluating $\eta$ and $\eta'$ at $2x$ rather than $x$. $F_{\mathrm{SP}}$ differs from $F_{\mathrm{NP}}$ by the additional terms $\delta\eta'(2x)[\,rx + (r^2/c_2)\eta(2x)\,]$, which are exactly the externality firm $i$'s period-1 data exerts on the rival's period-2 productivity, ignored by the pooled-Nash firm, internalized by the planner. \smallskip \\
Finally, let $x_{\mathrm{NF}}$, $x_{\mathrm{NP}}$, $x_{\mathrm{SP}}$ denote the symmetric first-period capacities solving $F_{\mathrm{NF}}=0$, $F_{\mathrm{NP}}=0$, $F_{\mathrm{SP}}=0$, respectively. Lemma~\ref{lem:period2} then gives the period-2 capacities
\[
    K^*_{\mathrm{NF},2}=x_{\mathrm{NF}}+\frac{r\eta(x_{\mathrm{NF}})}{c_2},\qquad
    K^*_{\mathrm{NP},2}=x_{\mathrm{NP}}+\frac{r\eta(2x_{\mathrm{NP}})}{c_2},\qquad
    K^*_{\mathrm{SP},2}=x_{\mathrm{SP}}+\frac{r\eta(2x_{\mathrm{SP}})}{c_2}.
\]
Total welfare in each regime is twice the per-firm reduced value at the equilibrium $x$, so we have
\begin{equation}\label{eq:welfare}
    W^*_{\mathrm{NF}}=2V_{\mathrm{NF}}(x_{\mathrm{NF}}),\qquad
    W^*_{\mathrm{NP}}=2V_{\mathrm{SP}}(x_{\mathrm{NP}}),\qquad
    W^*_{\mathrm{SP}}=2V_{\mathrm{SP}}(x_{\mathrm{SP}}).
\end{equation}
The pooled-Nash total welfare is written in terms of $V_{\mathrm{SP}}$ rather than $V_i^{\mathrm{NP}}$ because $V_i^{\mathrm{NP}}(x,x)=V_{\mathrm{SP}}(x)$ at symmetric play.

\subsection{Regularity and Welfare-Ladder Lemma}
\label{ssec:regularity}

\begin{assumption}[Regularity]\label{ass:regularity}
On the relevant capacity domain, $\eta(N)>0$. The reduced objectives $V_{\mathrm{NF}}$ and $V_{\mathrm{SP}}$ are strictly concave, and $V_i^{\mathrm{NP}}(\cdot,x_j)$ is strictly concave for every relevant $x_j$. The first-order conditions $F_{\mathrm{NF}}$, $F_{\mathrm{NP}}$, and $F_{\mathrm{SP}}$ are strictly decreasing and have unique interior roots $x_{\mathrm{NF}}, x_{\mathrm{NP}}, x_{\mathrm{SP}} \in (0,\infty)$.
\end{assumption}
\noindent Assumption~\ref{ass:regularity} is the only maintained assumption beyond the primitive smoothness, monotonicity, concavity, and boundedness of $\eta$ stated in the setup. It is satisfied whenever $c_1$ is large enough relative to the curvature of $\eta$ and it rules out non-uniqueness of the symmetric equilibrium in each regime.

\begin{lemma}[Welfare-ladder lemma]\label{lem:ladder}
Under Assumption~\ref{ass:regularity},
\[
    V_{\mathrm{SP}}(x_{\mathrm{NP}}) \;\geq\; V_{\mathrm{NF}}(x_{\mathrm{NF}}),
\]
with strict inequality whenever $\eta$ is strictly increasing on the relevant domain.
\end{lemma}

\subsection{Main Results}

\begin{proposition}[Dynamic welfare ladder]\label{prop:ladder}
Under Assumption~\ref{ass:regularity},
\[
    W^*_{\mathrm{NF}}\leq W^*_{\mathrm{NP}}\leq W^*_{\mathrm{SP}}.
\]
The first inequality is strict whenever $\eta$ is strictly increasing on the relevant domain. The second inequality is strict whenever $x_{\mathrm{NP}}\neq x_{\mathrm{SP}}$.
\end{proposition}

\begin{proposition}[Strategic underinvestment]\label{prop:underinvestment}
Under Assumption~\ref{ass:regularity},
\[
    K^*_{\mathrm{NP},1}\leq K^*_{\mathrm{SP},1},
    \qquad
    K^*_{\mathrm{NP},2}\leq K^*_{\mathrm{SP},2}.
\]
If $\eta'(N)>0$ on the relevant domain, both inequalities are strict.
\end{proposition}

\begin{proposition}[Capacity reversal sign test]\label{prop:reversal}
Under Assumption~\ref{ass:regularity}, define
\[
    D(x)
    \;=\;
    r\bigl[\eta(2x)-\eta(x)+x\bigl(\eta'(2x)-\eta'(x)\bigr)\bigr]
    +\frac{r^2}{c_2}\bigl[\eta(2x)\eta'(2x)-\eta(x)\eta'(x)\bigr].
\]
Then
\[
    K^*_{\mathrm{NP},1}>K^*_{\mathrm{NF},1}
    \quad\Longleftrightarrow\quad
    D(K^*_{\mathrm{NF},1})>0,
\]
and $K^*_{\mathrm{NP},1}=K^*_{\mathrm{NF},1}$ if and only if $D(K^*_{\mathrm{NF},1})=0$.
\end{proposition}

\noindent The sign test $D(x)$ collects two effects of doubling the period-2 learning state from $x$ (fragmented) to $2x$ (pooled). The first is the level effect. $r\bigl[\eta(2x)-\eta(x)\bigr]$ measures the lift in period-2 productivity on already-installed capacity. By monotonicity of $\eta$ this term is non-negative and it is strictly positive whenever $\eta$ is strictly increasing. There also is the slope effect. The remainder, $rx\bigl[\eta'(2x)-\eta'(x)\bigr] + (r^2/c_2)\bigl[\eta(2x)\eta'(2x)-\eta(x)\eta'(x)\bigr]$, measures how pooling changes the marginal return to firm $i$'s own period-1 capacity, via the dependence of $\phi$ on the learning state. Strict concavity of $\eta$ gives $\eta'(2x)\leq \eta'(x)$, so the first piece is non-positive and the second piece pairs a higher level $\eta(2x)\geq\eta(x)$ with a (weakly) lower slope $\eta'(2x)\leq\eta'(x)$, and can take either sign. \smallskip \\
The sign of $D(x)$ therefore turns on how persistently the learning curve keeps learning between $x$ and $2x$. If marginal learning persists ($\eta'(2x)$ close to $\eta'(x)$, learning curve still steep at the pooled state), the level effect dominates, $D(x)>0$, and pooled Nash invests more than fragmented Nash. The pool's productivity boost reinforces the marginal value of own period-1 capacity. If marginal learning decays sharply ($\eta'(2x)\ll\eta'(x)$, learning saturates quickly past $x$), the slope effect dominates, $D(x)<0$, and pooled Nash invests less than fragmented Nash. Once the rival's data has already lifted productivity to near saturation, an individual firm has little additional reason to deploy capacity for its data-generation role. \smallskip \\
The sign of $D$ is not pinned down by monotonicity and concavity of $\eta$ alone. It is a quantitative property of how learning is distributed across the relevant range. The mapping $K^*_{\mathrm{NP},1}\gtrless K^*_{\mathrm{NF},1}$ is, in this sense, an empirical question about the shape of the learning curve. 

\begin{proposition}[Welfare wedges]\label{prop:cs}
Define
\[
    \Delta_{\mathrm{pool}}=W^*_{\mathrm{NP}}-W^*_{\mathrm{NF}},
    \qquad
    \Delta_{\mathrm{coord}}=W^*_{\mathrm{SP}}-W^*_{\mathrm{NP}}.
\]
Under Assumption~\ref{ass:regularity}, we have
\begin{itemize}[leftmargin=2em, nosep]
\item[\textup{(a)}] $\Delta_{\mathrm{pool}}\geq 0$ and $\Delta_{\mathrm{coord}}\geq 0$.
\item[\textup{(b)}] If learning is constant ($\eta(N)\equiv\eta_{\min}$), then $\Delta_{\mathrm{pool}}=\Delta_{\mathrm{coord}}=0$.
\item[\textup{(c)}] If $\eta$ is strictly increasing on the relevant domain, then $\Delta_{\mathrm{pool}}>0$ and $\Delta_{\mathrm{coord}}>0$.
\end{itemize}
\end{proposition}

\noindent The wedges are signed but their global dependence on $(\delta, c_1, c_2)$ and on a one-dimensional learning intensity is not monotone in general, because changes in these parameters shift the three capacities $x_{\mathrm{NF}}, x_{\mathrm{NP}}, x_{\mathrm{SP}}$ in directions that need not align. The constant-learning point in part~(b) is the only configuration in which the wedges vanish identically, so any deformation away from constant learning opens both wedges. This is the substantive content of part~(c). A monotone global comparative-statics statement would require additional restrictions on the joint movement of the three capacities, which we do not impose.

\paragraph{Discussion.} The four propositions deliver a clean baseline. Proposition~\ref{prop:ladder} says that pooled regimes weakly dominate fragmented ones in welfare, with strict gain whenever the learning curve is strictly increasing. Proposition~\ref{prop:underinvestment} says that, even within a pooled regime, firms underinvest in period-1 capacity relative to the planner, because they ignore the rival's productivity gain from their own period-1 data. The period-2 gap is inherited from the period-1 wedge, with no direct period-2 externality in this two-period horizon. Proposition~\ref{prop:reversal} converts the comparison between pooled-Nash and fragmented-Nash period-1 capacity into a sign test on $D(K^*_{\mathrm{NF},1})$. The test is robust as a test but not as a sign, because the shape of $\eta$ over $[K^*_{\mathrm{NF},1}, 2K^*_{\mathrm{NF},1}]$ determines whether pooled Nash invests more or less. Proposition~\ref{prop:cs} closes the section by noting that both welfare wedges are positive under any non-trivial learning curve, with the constant-learning case as the only point at which they vanish. The robust results are the welfare ladder and pooled-Nash underinvestment. The capacity comparison and the magnitude of the wedges are sign tests whose answers depend on the persistence of the learning curve.


\section{Product-Market Competition with Endogenous Data Sharing}
\label{sec:endo-sharing}

The baseline of Section \ref{sec:simple} treats the per-unit revenue $r$ as exogenous and the data-sharing architecture as a fixed regime. This section relaxes both. Firms compete in a downstream Cournot market with endogenous price and the data-sharing decision is taken strategically before period-1 capacity is chosen. The two-period structure, the capacity-acquisition cost technology, and the lagged-learning specification of Section \ref{sec:simple} are retained. The main analysis uses a linear inverse demand for tractability. Section \ref{app:general-demand} establishes that the central result extends to a general inverse demand.

\subsection{Linear Inverse Demand}

\subsubsection{Setup}
\label{ssec:endo-setup}

\paragraph{Demand.} The inverse demand in each period is linear,
\begin{equation}
P_t(Q_t) = a - b\, Q_t, \qquad a > 0, \quad b > 0,
\end{equation}
with $Q_t = q_{1,t} + q_{2,t}$. We maintain full capacity utilization, $q_{i,t} = K_{i,t}\,\eta_{i,t}$, normalizing per-period operating costs to zero.\footnote{Any positive constant operating cost $w \geq 0$ can be absorbed into the inverse demand by replacing $a$ with $a - w$. The qualitative results are unaffected.} Linearity is adopted for tractability and clean interpretation. Section~\ref{app:general-demand} re-derives it for a general inverse demand $P(Q)$ and shows the threshold is governed by the elasticity of industry demand over the relevant output range.
 
\paragraph{Timing.} The extended game has the following five periods.
\begin{itemize}[itemsep=0pt, topsep=0.4em]
\item[(0)] \textbf{Period-0 data-sharing.} Firms simultaneously announce $s_i \in \{S, N\}$. The pool forms if and only if $s_1 = s_2 = S$.
\item[(1)] \textbf{Period-1 capacity.} Firms simultaneously choose $K_{i,1} \geq 0$.
\item[(2)] \textbf{Period-1 production.} Outputs $q_{i,1} = K_{i,1}\eta_{\min}$ are produced and the market clears at $P_1$.
\item[(3)] \textbf{Period-2 capacity.} Given $(K_{1,1}, K_{2,1})$, firms simultaneously choose $K_{i,2} \geq K_{i,1}$.
\item[(4)] \textbf{Period-2 production.} Productivities $\eta_{i,2}$ realize according to the period-0 outcome, outputs $q_{i,2} = K_{i,2}\eta_{i,2}$ are produced, and the market clears at $P_2$.
\end{itemize}

\paragraph{Sharing regime.} Period-2 productivity satisfies
\begin{equation}
\eta_{i,2} =
\begin{cases}
\eta(K_{1,1} + K_{2,1}) & \text{if } (s_1, s_2) = (S, S), \quad \text{(pooled)} \\
\eta(K_{i,1}) & \text{otherwise} \quad \quad \qquad \ \ \  \text{(fragmented)}.
\end{cases}
\end{equation}

\paragraph{Costs and profit.} The capacity-cost structure is unchanged from Section 2. The total cost is $(c_1 /2) K_{i,1}^2 + (c_2/2)(K_{i,2} - K_{i,1})^2$, with $c_1, c_2 > 0$ and $\delta \in (0,1]$. Firm $i$'s discounted two-period profit is
\begin{align}
\pi_i \;=\; & \bigl[a - b(K_{1,1} + K_{2,1})\eta_{\min}\bigr] K_{i,1}\eta_{\min} - \tfrac{c_1}{2} K_{i,1}^2 \nonumber \\
& \;+\; \delta \Bigl\{ \bigl[a - b(K_{1,2}\eta_{1,2} + K_{2,2}\eta_{2,2})\bigr] K_{i,2}\eta_{i,2} - \tfrac{c_2}{2}(K_{i,2} - K_{i,1})^2 \Bigr\}.
\end{align}

\subsubsection{Reduced-Form Characterization}

We first characterize the symmetric period-2 Cournot equilibrium taking productivities and period-1 capacities as given.

\begin{lemma}[Period-2 Cournot best response]
\label{lem:p2-cournot}
Fix $(K_{1,1}, K_{2,1})$ and $(\eta_1, \eta_2)$. Firm $i$'s period-2 best response is
\begin{equation}
K_{i,2}^{\mathrm{br}}(K_{j,2}; \eta_i, \eta_j, K_{i,1}) \;=\; \max\!\left\{K_{i,1}, \; \frac{\eta_i(a - b\eta_j K_{j,2}) + c_2 K_{i,1}}{2b\eta_i^2 + c_2}\right\}.
\end{equation}
At a symmetric continuation with $\eta_1 = \eta_2 = \eta$ and $K_{1,1} = K_{2,1} = K_1$, the unique symmetric Cournot equilibrium has
\begin{equation}
K_2^*(\eta, K_1) \;=\; \frac{a\eta + c_2 K_1}{3b\eta^2 + c_2},
\end{equation}
provided $a > 3 b \eta K_1$, in which case the irreversibility constraint is slack.
\end{lemma}

\noindent Substituting $K_2^*$ into the period-2 payoff defines the symmetric reduced period-2 value
\begin{equation}
\label{eq:v2-def}
v_2(\eta, K_1) \;:=\; K_2^*(\eta, K_1)\, \eta \bigl[a - 2 b \eta K_2^*(\eta, K_1)\bigr] \;-\; \tfrac{c_2}{2}\bigl[K_2^*(\eta, K_1) - K_1\bigr]^2.
\end{equation} 
At a symmetric period-1 continuation with $K_{1,1} = K_{2,1} = K$, and given the period-0 outcome $R \in \{P, F\}$, we write
\begin{equation}
\eta^P(K) \;=\; \eta(2K), \qquad \eta^F(K) \;=\; \eta(K),
\end{equation}
and define the symmetric reduced two-period payoff
\begin{equation}
V^R(K) \;:=\; \bigl(a - 2 b \eta_{\min} K\bigr) K \eta_{\min} \;-\; \tfrac{c_1}{2} K^2 \;+\; \delta \, v_2\bigl(\eta^R(K), K\bigr), \qquad R \in \{P, F\}.
\end{equation}
For each regime, the symmetric period-1 Cournot--Nash first-period capacity is the unique interior root (when it exists) of
\begin{equation}
\label{eq:stage1-foc}
\left.\frac{\partial}{\partial K_i} \widetilde V^R(K_i, K_j) \right|_{K_i = K_j = K} \;=\; 0,
\end{equation}
where $\widetilde V^R(K_i, K_j)$ is the asymmetric two-period payoff to firm $i$ under regime $R$ before symmetry is imposed. We denote the resulting symmetric Period-1 capacities by $K^P_\star$ and $K^F_\star$, but do not characterize these in closed form.

\paragraph{Period-0 payoffs.} Given the continuation equilibrium, the period-0 normal-form game has payoffs
\begin{equation}
U^P \;:=\; V^P(K^P_\star), \qquad U^F \;:=\; V^F(K^F_\star),
\end{equation}
with $U^P$ obtained at $(S, S)$ and $U^F$ obtained at each of the other three action profiles. The profile $(N, N)$ is a Nash equilibrium for any parameter values. The profile $(S, S)$ is also a Nash equilibrium whenever $U^P \geq U^F$. Following the convention that the Pareto-dominant Nash is selected when available, we say that \emph{pooling is sustainable} when $U^P > U^F$.

\subsubsection{Exogenous Period-1 Capacity}
\label{ssec:threshold}
To extract a closed-form threshold, we consider a benchmark in which period-1 capacity is fixed at a common predetermined fleet size $K>0$. This captures settings in which the period-1 deployment scale is determined by binding capital constraints, supply constraints, procurement targets, or an installed base, rather than by the firm's optimization. In this benchmark, the period-1 component of $V^R(K)$ is identical across regimes, and the period-0 comparison reduces to
\begin{equation}
\label{eq:stage0-reduced}
U^P - U^F \;=\; \delta \, \bigl[ v_2(\eta^P, K) - v_2(\eta^F, K) \bigr], \qquad \eta^P = \eta(2K), \;\; \eta^F = \eta(K).
\end{equation}

\begin{assumption}[Interior period-2 Cournot]
\label{ass:interior-cournot}
At the predetermined fleet size $K$, the demand slope $b$ satisfies $b < a / (3 K \eta^P)$, so that the period-2 irreversibility constraint is slack at the symmetric Cournot equilibrium in both regimes.
\end{assumption}

\begin{proposition}[Sharing-sustainability threshold]
\label{prop:threshold}
Under Assumption~\ref{ass:interior-cournot} and fixed period-1 capacity $K > 0$:
\begin{itemize}[itemsep=0pt, topsep=0.4em]
\item[\textup{(a)}] In the further limit $c_2 \to \infty$ (no period-2 expansion), pooling is sustainable if and only if
\begin{equation}
\label{eq:b-star}
b \;<\; b^\star(K) \;:=\; \frac{a}{2 K\,(\eta^P + \eta^F)}.
\end{equation}
\item[\textup{(b)}] For finite $c_2 > 0$, the map $b \mapsto U^P(b) - U^F(b)$ is strictly positive as $b \to 0^+$, and is strictly decreasing in $b$ at $b = 0$. Its sign away from $b = 0$ is governed by the monotonicity of $v_2(\cdot, K)$ in $\eta$, which admits the envelope decomposition
\begin{equation}
\label{eq:envelope}
\frac{\partial v_2}{\partial \eta}
\;=\; \underbrace{K_2^*\bigl(a - 4 b \eta K_2^*\bigr)}_{\text{direct effect}}
\;-\; \underbrace{b\,\eta^2 K_2^*\,\frac{\partial K_2^*}{\partial \eta}}_{\text{strategic effect}},
\end{equation}

\end{itemize}
\end{proposition}

\noindent The closed-form threshold in part~(a) requires both that period-1 capacity be exogenously fixed and that period-2 expansion vanish in the limit as $c_2\to\infty$, and relaxing either limit loses closed form. With $c_2 < \infty$ (part~(b)), the period-2 Cournot equilibrium makes $v_2$ depend on $b$ through $K_2^*(\eta, K)$. Two local facts survive: the limit $b \to 0^+$ flattens demand and recovers the exogenous-price baseline, where pooling is beneficial; and $U^P - U^F$ is strictly decreasing in $b$ at $b = 0$, so competition erodes the gain from pooling as soon as it is introduced. Away from $b = 0$, the envelope decomposition~\eqref{eq:envelope} isolates a direct effect, a firm's gain from its own higher productivity, and a strategic effect, the price-mediated impact of the rival's equilibrium capacity response. Neither is signed over the whole interval, which is why the finite-$c_2$ threshold admits no closed form. With $c_1 < \infty$, the symmetric period-1 first-order conditions differ across regimes, so the equilibrium fleet sizes $K^P_\star$ and $K^F_\star$ enter the period-0 comparison non-trivially. Section~\ref{ssec:wedge} characterizes the resulting capacity wedge in the limit where period-2 expansion vanishes via a sign test, without resolving its sign unconditionally. \smallskip \\
The threshold $b^\star(K)$ in~\eqref{eq:b-star} is decreasing in both $K$ and the productivity levels $(\eta^P, \eta^F)$. $a$ measures how much output the market can absorb, while the denominator reflects the scale of aggregate output across the pooled and fragmented productivity regimes. Pooling raises productivity and so joint output, so it is sustainable only when the market is large enough to absorb the extra output without the price collapsing. Industries with larger fleets, or already at higher productivity, therefore have a lower sustainability threshold, a comparative-statics statement about the benchmark only.

\subsubsection{Endogenous Period-1 Capacity}
\label{ssec:wedge}
 
We now consider the complementary benchmark in which period-1 capacity is chosen endogenously with finite cost \(c_1>0\), while period-2 capacity expansion vanishes, \(c_2\to\infty\). This isolates the period-1 investment decision, the channel switched off in Proposition~\ref{prop:threshold}. \smallskip \\
With $c_2 \to \infty$ the period-2 irreversibility constraint binds, so $K_{i,2} = K_{i,1}$ and we write $K_i := K_{i,1} = K_{i,2}$ for the single capacity of firm $i$. Firm $i$ produces $q_{i,1} = K_i\eta_{\min}$ in period~1 and $q_{i,2} = K_i\eta_{i,2}$ in period~2, with two-period profit under regime $R \in \{P, F\}$
\begin{equation}
\label{eq:wedge-profit}
\pi_i^R \;=\; \bigl[a - b\eta_{\min}(K_i + K_j)\bigr] K_i\eta_{\min} - \tfrac{c_1}{2} K_i^2 + \delta\bigl[a - b(K_i\eta_{i,2} + K_j\eta_{j,2})\bigr] K_i\eta_{i,2},
\end{equation}
where $\eta_{i,2} = \eta_{j,2} = \eta(K_i + K_j)$ under pooling and $\eta_{i,2} = \eta(K_i)$ under fragmentation. Differentiating \eqref{eq:wedge-profit} with respect to $K_i$, holding $K_j$ fixed, and evaluating at $K_i = K_j = K$, the symmetric period-1 Cournot--Nash capacity in regime $R$ solves
\begin{equation}
\label{eq:wedge-foc}
F^R(K) \;:=\; a\eta_{\min} - (3 b \eta_{\min}^2 + c_1)\,K + \delta g^R(K) \;=\; 0,
\end{equation}
where $g^R(K)$ is the symmetric evaluation of the derivative of the period-2 term,
\begin{align}
g^F(K) &= \bigl(\eta(K) + K\eta'(K)\bigr)\bigl(a - 3 b K \eta(K)\bigr), \\
g^P(K) &= \eta(2K)\bigl(a - 3 b K \eta(2K)\bigr) + K\eta'(2K)\bigl(a - 4 b K \eta(2K)\bigr).
\end{align}
 
\begin{assumption}[Period-1 regularity]
\label{ass:stage1-reg}
In the limit $c_2 \to \infty$ with $c_1 \in (0, \infty)$, for each $R \in \{P, F\}$ the function $F^R$ is strictly decreasing on $(0, \infty)$ and has a unique interior root $K^R_\star > 0$.
\end{assumption}
\noindent Assumption~\ref{ass:stage1-reg} requires the period-1 investment problem to be well-behaved (single-peaked). It holds, in particular, when the period-1 capacity cost is sufficiently convex relative to the strength of the learning channel, so that the marginal cost of capacity dominates any region of locally increasing marginal returns generated by the Cournot price interaction. It rules out multiplicity of period-1 capacities, which can arise when capacity is nearly costless and the learning curve has a sharp convex--concave transition, but we do not study that regime here.
 
\begin{proposition}[Capacity-wedge sign test]
\label{prop:wedge}
Under Assumption~\ref{ass:stage1-reg}, the symmetric period-1 capacities satisfy: $K^P_\star$ exceeds, equals, or falls short of $K^F_\star$ according as $\Delta(K^F_\star)$ is positive, zero, or negative, where at any common capacity $K$,
\begin{align}
\Delta(K) &= \Delta_0(K) + \Delta_{\mathrm{comp}}(K), \\
\Delta_0(K) &= a\Bigl[\bigl(\eta(2K) - \eta(K)\bigr) + K\bigl(\eta'(2K) - \eta'(K)\bigr)\Bigr], \\
\Delta_{\mathrm{comp}}(K) &= -3 b K\bigl(\eta(2K)^2 - \eta(K)^2\bigr) - b K^2\bigl(4\,\eta(2K)\eta'(2K) - 3\,\eta(K)\eta'(K)\bigr).
\end{align}
Moreover, we have 
\begin{itemize}[itemsep=0pt, topsep=0.4em]
\item[\textup{(i)}] $\Delta_0$ coincides with the $c_2 \to \infty$ limit of the exogenous-price capacity wedge $D$ in Proposition~\ref{prop:reversal} of the baseline, with $a$ in the role of $r$.Its sign is not determined by the monotonicity and concavity of $\eta$ alone.
\item[\textup{(ii)}] The term $-3 b K(\eta(2K)^2 - \eta(K)^2)$ in $\Delta_{\mathrm{comp}}$ is strictly negative.
\item[\textup{(iii)}] If $4\,\eta(2K)\eta'(2K) \geq 3\,\eta(K)\eta'(K)$, then $\Delta_{\mathrm{comp}}(K) < 0$, and hence $\Delta(K) < \Delta_0(K)$.
\end{itemize}
\end{proposition}

\noindent Proposition~\ref{prop:wedge} is a sign test, not a sign. Whether $K^P_\star$ exceeds or falls short of $K^F_\star$ depends on $\Delta(K^F_\star)$, and through $\Delta_0$ on the shape of the learning curve, exactly as in the exogenous-price baseline. What is new under Cournot competition is $\Delta_{\mathrm{comp}}$, absent when the price is exogenous. It collects the price-mediated effect of pooled learning, since under pooling a firm's period-1 capacity raises the rival's productivity and so the rival's period-2 output. Part~(ii) shows one component of $\Delta_{\mathrm{comp}}$ is unambiguously negative, and part~(iii) gives a transparent condition under which $\Delta_{\mathrm{comp}}$ is negative in full, so that the Cournot wedge lies strictly below the exogenous-price wedge. Competition then shifts the capacity comparison toward pooled under-investment relative to the baseline. This does not by itself establish $K^P_\star < K^F_\star$ for any given economy, since the sign of $\Delta$ still depends on $\Delta_0$. \smallskip \\
The condition in part~(iii), $4\,\eta(2K)\eta'(2K) \ge 3\,\eta(K)\eta'(K)$, concerns the persistence of the learning slope rather than its magnitude. A simple sufficient condition is $\eta'(2K) \ge (3/4)\eta'(K)$. The exact condition is weaker, since it only requires $\eta'(2K) \ge \frac{3\eta(K)}{4\eta(2K)}\eta'(K)$. A learning curve that remains steep across $[K, 2K]$ delivers $\Delta_{\mathrm{comp}}(K) < 0$, whereas one that saturates sharply may violate the condition and leave its sign indeterminate. What matters is then not how fast an industry learns but how long it keeps learning. The competitive cost of sharing, that pooled data strengthens the rival, is large and persistent where learning stays steep, and weak where it saturates early, so an industry that learns quickly but saturates sharply can carry a low competitive cost despite a high peak learning rate.
\paragraph{Discussion.} The analysis qualifies the baseline's policy prescription. The baseline yields an unambiguous rule. Pooled regimes underinvest in early adoption (Proposition~\ref{prop:underinvestment} of the baseline), so a planner should subsidize period-1 deployment. The competitive correction $\Delta_{\mathrm{comp}}$ is a force that rule omits, and under the condition above it pushes the wedge toward $K^P_\star < K^F_\star$, the more so the more intense is downstream competition. The direction of the optimal intervention is therefore not a primitive. In weakly competitive industries the case for subsidizing early deployment stands, whereas in strongly competitive ones the same data-pooling architecture can blunt the incentive to deploy early, so a uniform early-adoption subsidy may be misdirected. Moreover, the welfare ordering of the baseline does not carry over. It rests on Lemma~\ref{lem:ladder} of the baseline, whose transparent sufficient condition, per-firm pooled value weakly exceeds per-firm fragmented value at every common capacity, fails under Cournot competition whenever $b > b^\star(K)$, since then $v_2(\eta^P, K) < v_2(\eta^F, K)$. We therefore do not extend the welfare ladder to the Cournot setting, and leave open whether \(W^P \gtrless W^F\) outside the benchmark with exogenous period-1 capacity and vanishing period-2 expansion.

\subsection{General Inverse Demand}
\label{app:general-demand}

This section shows that the closed-form sustainability threshold of Proposition~\ref{prop:threshold}(a) does not depend on the linear demand specification. Let $P : \mathbb{R}_{>0} \to \mathbb{R}_{>0}$ be a generic twice-differentiable inverse demand, and let $T(Q) := Q\,P(Q)$ denote total industry revenue. All other elements of the model in Section~\ref{ssec:endo-setup} are retained, including full capacity utilization, which we maintain whenever the symmetric Cournot marginal revenue is non-negative at the relevant aggregate output. \smallskip \\
In the benchmark with exogenous period-1 capacity and no period-2 expansion, each firm operates a fixed capacity \(K\) in both periods, so \(K_{i,1}=K_{i,2}=K\). Aggregate period-2 output at a symmetric continuation with common productivity \(\eta\) is therefore \(Q_2(\eta)=2K\eta\), and the per-firm reduced value is
\begin{equation}
v_2(\eta, K) \;=\; P(2K\eta) \cdot K\eta \;=\; \tfrac{1}{2}\, T(2K\eta).
\end{equation}
Substituting into \eqref{eq:stage0-reduced},
\begin{equation}
\label{eq:general-sustainability}
U^P - U^F \;=\; \tfrac{\delta}{2}\,\bigl[\, T(2K\eta^P) - T(2K\eta^F)\,\bigr].
\end{equation}

\begin{proposition}[General sustainability characterization]
\label{prop:general-threshold}
Under the assumptions above, pooling is sustainable if and only if $T(2K\eta^P) > T(2K\eta^F)$.
\end{proposition}

\noindent The standard identity $T'(Q) = P(Q)\bigl(1 - 1/\varepsilon_D(Q)\bigr)$, where $\varepsilon_D(Q) := -P(Q)/[Q\,P'(Q)]$ is the local demand elasticity, sharpens Proposition~\ref{prop:general-threshold}. Pooling is sustainable when $\varepsilon_D > 1$ throughout $[2K\eta^F,\, 2K\eta^P]$ and is not sustainable when $\varepsilon_D < 1$ throughout that interval. If $\varepsilon_D$ crosses unity inside it, the question is settled by the direct comparison in Proposition~\ref{prop:general-threshold}. The two leading specifications recover this. Linear inverse demand $P(Q) = a - bQ$ gives the threshold $b < a/[2K(\eta^P + \eta^F)]$ of Proposition~\ref{prop:threshold}(a), which inherits dependence on $K$ and $(\eta^P, \eta^F)$ because the local elasticity varies with $Q$ and constant-elasticity demand $P(Q) = A\,Q^{-1/\varepsilon}$ gives the parameter-free threshold $\varepsilon > 1$, because the elasticity is constant. 
\paragraph{Discussion.} The policy content is that the sustainability of data pooling is governed by demand elasticity over the output interval induced by the two productivity regimes $[2K\eta^F,\,2K\eta^P]$, rather than by the productivity parameters in isolation. Pooling raises both firms' productivity and so aggregate output. Whether this raises or lowers total industry revenue, and so whether sharing is self-sustaining, depends only on whether demand is elastic or inelastic over the relevant output range. The unit-elasticity line $\varepsilon_D = 1$ is the policy boundary. In elastic industries the standard prescription of encouraging data pooling applies, while in inelastic industries the same architecture is self-defeating. This also flags a counterintuitive case for policymakers. Necessity-type goods, whose demand is typically inelastic, are precisely the industries where mandated data sharing is least likely to be beneficial. We do not claim Proposition~\ref{prop:general-threshold} extends to finite $c_2$ or to endogenous period-1 capacity. With $c_2 < \infty$ the period-2 Cournot equilibrium $K_2^*$ depends on $P$ through the firm's first-order condition, and the reduction to a one-line comparison of $T(\cdot)$ no longer holds.

\section{Numerical Exploration under General Demand}
\label{sec:numerical-general}

The closed-form results of Section~\ref{sec:endo-sharing} rest on a linear inverse demand and on shutting down either period-1 investment or period-2 expansion. This section removes those simplifications. We solve the full period-0/1/2 game numerically for an arbitrary inverse demand $P(Q)$, with period-1 capacity chosen endogenously and period-2 expansion active. The aim is not to verify a theorem but to see which qualitative patterns survive once the convenient assumptions are dropped. \smallskip \\
We report the following three representative demand shapes, each standing for a different
kind of market.
\begin{itemize}[itemsep=0pt, topsep=0.4em]
\item \textbf{Isoelastic (elastic market).} Buyers are price-sensitive and the market is ``deep''. Extra output is absorbed with only a small price concession, so industry revenue keeps rising with output. This is the growth-stage, far-from-saturation case.
\item \textbf{Logit (maturing market).} An S-shaped, saturating market. Output can still expand early, but each additional unit depresses price more sharply as the ceiling approaches. This is the realistic intermediate case of an industry heading toward maturity.
\item \textbf{Convex (capacity-limited market).} Price collapses quickly as output grows and the market cannot absorb much. This is the necessity-type case in which extra supply mainly destroys price.
\end{itemize}

\subsection{Equilibrium Purchasing Paths}
Figure~\ref{fig:sim-paths} traces the symmetric equilibrium capacity in each period, under pooling ($P$) versus fragmentation ($F$), for two competing firms. The main message is that the effect of pooling on early purchasing has no universal sign, and it is set by the shape of demand. In the elastic and maturing markets, pooling slightly lowers period-1 capacity, as each firm leans on the rival's contribution to the shared learning stock. In the capacity-limited market the comparison reverses. Pooling raises early purchasing, yet sharing is no longer self-sustaining. In every case the path is back-loaded ($K_2 > K_1$). Firms deploy in period~1 partly to seed the learning stock, then expand in period~2 once higher productivity is realized.

\begin{figure}[h!]
\centering
\includegraphics[width=\textwidth]{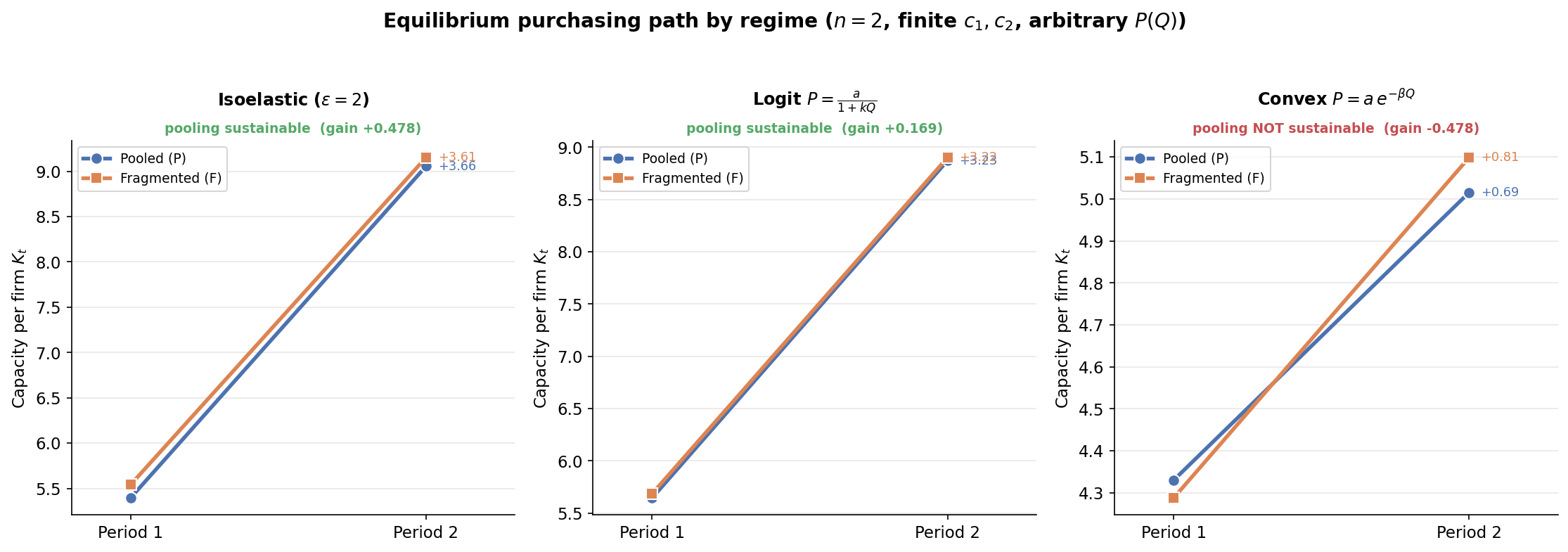}
\caption{Symmetric equilibrium capacity in period~1 and period~2 under pooled
($P$) and fragmented ($F$) regimes, for three inverse-demand shapes ($n=2$,
finite $c_1, c_2$). Whether pooling raises or lowers early deployment depends
on the demand shape, and flips sign between the elastic/maturing and the capacity-limited
cases.}
\label{fig:sim-paths}
\end{figure}

\subsection{The Role of Competition}

Figure~\ref{fig:sim-competition} varies the number of competing firms $n$. The left panel plots the pooling gain $U^P - U^F$ and the right panel plots the period-1 capacity wedge $K_1^P - K_1^F$, i.e.\ whether pooling raises or cuts early purchasing. Two high-level points emerge. First, the claim that competition necessarily erodes the value of pooling is not a universal law. The gain stays positive and even rises slightly in the elastic market, erodes gradually in the maturing market, and is negative throughout in the capacity-limited market. The direction in which competition bites is itself governed by demand elasticity. Second, the early-purchasing wedge inherits the same split, pooling depresses early deployment where the market is elastic or maturing, and lifts it where the market is capacity-limited, with both effects settling to a plateau as the market crowds.

\begin{figure}[h!]
\centering
\includegraphics[width=\textwidth]{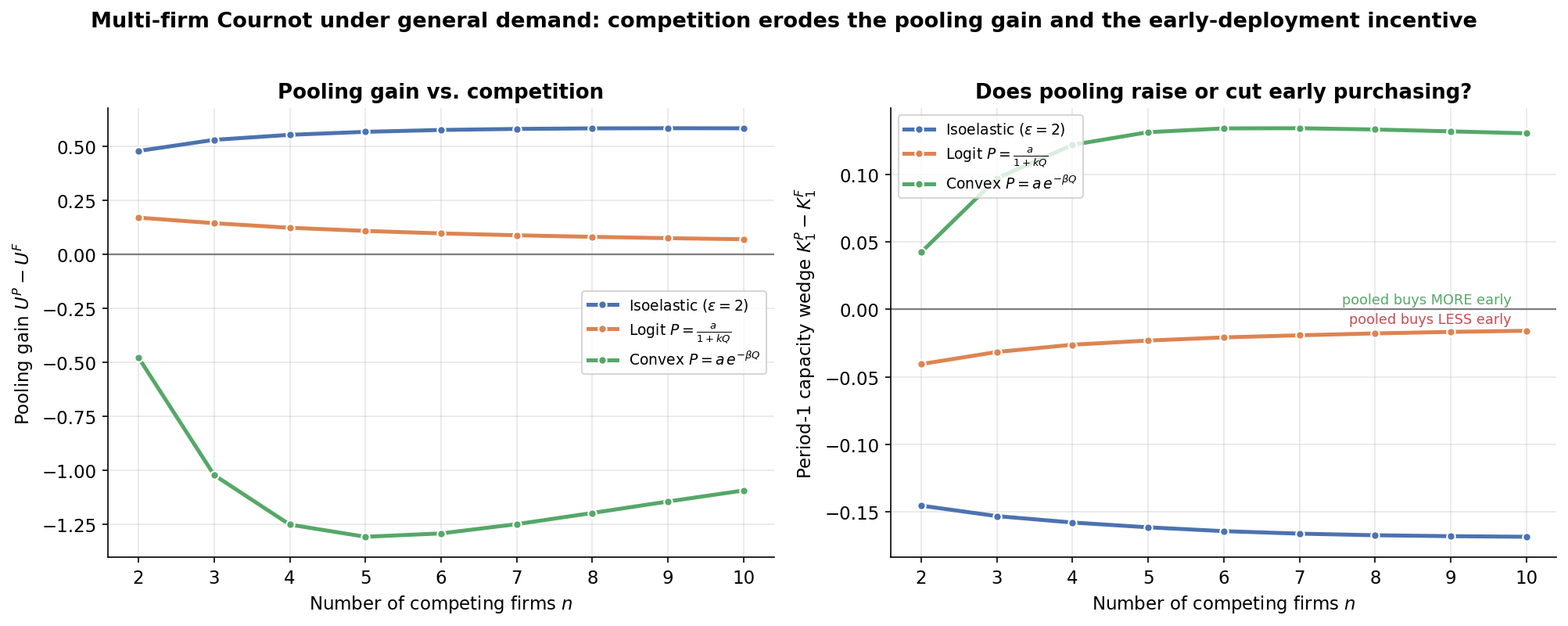}
\caption{Effect of competition under general demand. Left: pooling gain
$U^P - U^F$ as the number of firms $n$ grows. Right: period-1 capacity wedge
$K_1^P - K_1^F$. The sign and the competitive comparative static both depend on
the demand shape rather than on a single threshold.}
\label{fig:sim-competition}
\end{figure}

\noindent Together the two figures refine the closed-form analysis. Competition intensity remains a first-order conditioning variable for data-pooling policy, but under general demand its effect is mediated by the elasticity of the market. The same pooling architecture is welfare-enhancing in deep, elastic industries and self-defeating in capacity-limited ones, with the maturing case in between. There is accordingly no uniform prescription for early-deployment subsidies or mandated interoperability. The right intervention depends jointly on how elastic demand is and on how many firms compete.

\section{Future Directions}

\paragraph{Heterogeneous downstream firms.}
The current model assumes symmetric adopting firms. A natural extension is to allow firms to differ in deployment costs, installed bases, demand exposure, or data-generating efficiency. For example, firm \(i\) could have cost parameters \((c_{1i},c_{2i})\) and data productivity \(\alpha_i\), so that the pooled learning stock becomes
\[
    N_2=\sum_i \alpha_i K_{i,1},
\]
while fragmented learning is \(N_{i,2}=\alpha_i K_{i,1}\). This would make the distribution of early deployment matter, rather than only aggregate capacity. Low-cost or data-rich firms may generate most of the learning spillover, while high-cost firms may prefer to free ride on the pool. Studying this heterogeneity would clarify when data sharing redistributes value across firms, when asymmetric firms voluntarily sustain pooling, and whether policy should target subsidies toward firms with the largest learning externalities rather than toward the industry as a whole.
\paragraph{Vendor as a third strategic player.}
The current model treats the platform vendor as a passive supplier. A natural extension introduces a profit-maximizing vendor that sets unit prices across periods. Firms would pay the vendor price for each unit acquired, in addition to firm-side integration costs governed by $c_1$ and $c_2$. If the vendor sets prices $p_1$ and $p_2$, its profit is
\begin{equation}
    \label{eq:vendor-profit}
    \Pi_V(p_1,p_2)
    = (p_1-\kappa_V)\sum_{i=1}^{2} K_{i,1}(p_1,p_2)
    + (p_2-\kappa_V)\sum_{i=1}^{2}
      \bigl[K_{i,2}(p_1,p_2)-K_{i,1}(p_1,p_2)\bigr],
\end{equation}
where $\kappa_V$ is the vendor's marginal production cost and $K_{i,2}(p_1,p_2)-K_{i,1}(p_1,p_2)$ is firm~$i$'s period-2 incremental capacity purchase. This pricing problem is dynamic. A lower period-1 price can induce greater early deployment, generate more learning data, raise period-2 productivity, and increase later willingness to pay. Under sufficiently strong learning effects, a forward-looking vendor may therefore price early units below the level chosen by a static vendor. This suggests that the vendor may partially internalize the learning externality ignored by downstream adopters, although the alignment between vendor incentives and the social planner is generally incomplete. A commitment problem may also arise. The vendor would like to commit to low early prices and higher later prices, but without commitment firms may anticipate future markups and reduce early adoption. Characterizing this pricing game would clarify when private platform incentives can substitute for, or complement, early-deployment subsidies.

\paragraph{Multiple vendors, multiple firms, and a mean-field formulation.}
A second extension allows for $V$ competing vendors and $F$ adopting firms. Each firm chooses a vendor and a capacity path, while each vendor sets prices and operates its own data pool. Let $m_v$ denote the mass of firms using vendor~$v$, and let $N_v$ denote the corresponding learning stock. In a large-market limit with $F\to\infty$, the model can potentially be approximated by a mean-field formulation in which each individual firm responds to the aggregate state
\[
    \bigl\{(m_v,\,N_v,\,p_v)\bigr\}_{v=1}^{V}.
\]
The learning stock on platform~$v$ would evolve according to the aggregate deployment of firms choosing that platform, for example
\begin{equation}
    \label{eq:platform-learning}
    N_{v,2} = \int_{\mathcal{I}_v} K_{i,1}\,di,
\end{equation} where $\mathcal{I}_v$ is the set of firms adopting platform~$v$. This formulation would make it possible to study equilibrium platform shares, price paths, and learning trajectories in markets with many adopters. Several conjectures naturally arise. First, cumulative learning may generate tipping dynamics, where small early differences in platform share become persistent advantages. Second, welfare need not be monotone in the number of vendors. More vendors increase product-market and platform competition, but they may also fragment deployment data and slow learning. Third, cross-platform data interoperability may recover part of the welfare gain from pooling without requiring a monopoly platform. A mean-field version of the model would also allow comparative statics with respect to firm-size distributions and industry concentration, which are difficult to analyze in the two-firm setting.

\bibliographystyle{plainnat}
\bibliography{ref}

\newpage

\appendix
\section{Proofs}
\label{app:proofs}

\subsection{Proofs of Lemmas}
\subsubsection{Proof of Lemma~\ref{lem:period2}}
\label{app:proof-lem-period2}
\begin{proof}
Let $y=K_{i,2}-K_{i,1}\geq 0$. The problem becomes $\max_{y\geq 0}\; r(K_{i,1}+y)\eta(N)-(c_2/2)y^2$. The objective is strictly concave in $y$; the unconstrained first-order condition gives $y^*=r\eta(N)/c_2 > 0$, so the irreversibility constraint is slack.
\end{proof}
\subsubsection{Proof of Lemma~\ref{lem:ladder}}
\label{app:proof-lem-ladder}
\begin{proof}
By the Nash condition, firm $i$'s best response to a rival playing $x_{\mathrm{NP}}$ is to play $x_{\mathrm{NP}}$ itself, so
\[
    V_i^{\mathrm{NP}}(x_{\mathrm{NP}},x_{\mathrm{NP}}) \;\geq\; V_i^{\mathrm{NP}}(x_{\mathrm{NF}},x_{\mathrm{NP}}).
\]
The left-hand side equals $V_{\mathrm{SP}}(x_{\mathrm{NP}})$ by the identity $V_i^{\mathrm{NP}}(x,x)=V_{\mathrm{SP}}(x)$. For the right-hand side, definition~\eqref{eq:V-NP} gives
\[
    V_i^{\mathrm{NP}}(x_{\mathrm{NF}},x_{\mathrm{NP}})
    \;=\; r\eta_{\min}\,x_{\mathrm{NF}}
    - \tfrac{c_1}{2}\,x_{\mathrm{NF}}^2
    + \delta\,\phi\bigl(x_{\mathrm{NF}},\,x_{\mathrm{NF}}+x_{\mathrm{NP}}\bigr).
\]
Since $\phi(K,\cdot)$ is non-decreasing and $x_{\mathrm{NF}}+x_{\mathrm{NP}} \geq x_{\mathrm{NF}}$, we have
\[
    \phi\bigl(x_{\mathrm{NF}},\,x_{\mathrm{NF}}+x_{\mathrm{NP}}\bigr)
    \;\geq\; \phi\bigl(x_{\mathrm{NF}},\,x_{\mathrm{NF}}\bigr),
\]
strict if $\eta$ is strictly increasing and $x_{\mathrm{NP}}>0$. So, it follows that
\[
    V_i^{\mathrm{NP}}(x_{\mathrm{NF}},x_{\mathrm{NP}})
    \;\geq\; r\eta_{\min}\,x_{\mathrm{NF}} - \tfrac{c_1}{2}\,x_{\mathrm{NF}}^2 + \delta\,\phi(x_{\mathrm{NF}},\,x_{\mathrm{NF}})
    \;=\; V_{\mathrm{NF}}(x_{\mathrm{NF}}).
\]
Chaining the two inequalities gives $V_{\mathrm{SP}}(x_{\mathrm{NP}})\geq V_{\mathrm{NF}}(x_{\mathrm{NF}})$, strict under strict monotonicity of $\eta$.
\end{proof}

\subsubsection{Proof of Lemma~\ref{lem:p2-cournot}}
\label{app:proof-lem-p2-cournot}
\begin{proof}
The unconstrained period-2 first-order condition is
\[
\eta_i \bigl(a - 2 b \eta_i K_{i,2} - b \eta_j K_{j,2}\bigr) \;=\; c_2(K_{i,2} - K_{i,1}),
\]
which, after solving for $K_{i,2}$ and imposing $K_{i,2} \geq K_{i,1}$, yields the displayed best-response. Setting $K_{1,2} = K_{2,2} = K_2$ at symmetry gives $\eta(a - 3 b \eta K_2) = c_2(K_2 - K_1)$, and rearrangement yields the formula. The constraint is slack at the symmetric solution if and only if $K_2^* > K_1$, equivalently $a > 3 b \eta K_1$.
\end{proof}

\subsection{Proofs of Propositions}
\subsubsection{Proof of Proposition~\ref{prop:ladder}}
\label{app:proof-prop-ladder}
\begin{proof}
By \eqref{eq:welfare}, $W^*_{\mathrm{NF}}=2V_{\mathrm{NF}}(x_{\mathrm{NF}})$ and $W^*_{\mathrm{NP}}=2V_{\mathrm{SP}}(x_{\mathrm{NP}})$. The welfare-ladder lemma (Lemma~\ref{lem:ladder}) gives $V_{\mathrm{SP}}(x_{\mathrm{NP}})\geq V_{\mathrm{NF}}(x_{\mathrm{NF}})$, so $W^*_{\mathrm{NP}}\geq W^*_{\mathrm{NF}}$, strict whenever $\eta$ is strictly increasing on the relevant domain. For the second inequality, the planner maximizes $V_{\mathrm{SP}}$ over a symmetric feasible set that contains $x_{\mathrm{NP}}$, so $V_{\mathrm{SP}}(x_{\mathrm{NP}})\leq V_{\mathrm{SP}}(x_{\mathrm{SP}})$ and $W^*_{\mathrm{NP}}\leq W^*_{\mathrm{SP}}$. Strict concavity of $V_{\mathrm{SP}}$ gives strict inequality whenever $x_{\mathrm{NP}}\neq x_{\mathrm{SP}}$.
\end{proof}

\subsubsection{Proof of Proposition~\ref{prop:underinvestment}}
\label{app:proof-prop-underinvestment}
\begin{proof}
At any common first-period capacity $x$, we have
\[
\begin{aligned}
    F_{\mathrm{SP}}(x)-F_{\mathrm{NP}}(x)
    &=
    \delta\left[
        r x\eta'(2x)
        +\frac{r^2}{c_2}\eta(2x)\eta'(2x)
    \right].
\end{aligned}
\]
This term is weakly positive and it is strictly positive when $\eta'(2x)>0$. Since $F_{\mathrm{NP}}(x_{\mathrm{NP}})=0$, we have $F_{\mathrm{SP}}(x_{\mathrm{NP}}) = F_{\mathrm{SP}}(x_{\mathrm{NP}})-F_{\mathrm{NP}}(x_{\mathrm{NP}}) \geq 0$. Because $F_{\mathrm{SP}}$ is strictly decreasing and has the unique root $x_{\mathrm{SP}}$, this implies $x_{\mathrm{NP}}\leq x_{\mathrm{SP}}$, with strict inequality when $\eta'>0$ on the relevant domain. So, $K^*_{\mathrm{NP},1}\leq K^*_{\mathrm{SP},1}$. For period 2, Lemma~\ref{lem:period2} gives
\[
    K^*_{\mathrm{NP},2}=x_{\mathrm{NP}}+\frac{r\eta(2x_{\mathrm{NP}})}{c_2},
    \qquad
    K^*_{\mathrm{SP},2}=x_{\mathrm{SP}}+\frac{r\eta(2x_{\mathrm{SP}})}{c_2}.
\]
The map $g(x)=x+r\eta(2x)/c_2$ is weakly increasing and strictly increasing if $\eta'\geq 0$. Since $x_{\mathrm{NP}}\leq x_{\mathrm{SP}}$, we obtain $K^*_{\mathrm{NP},2}\leq K^*_{\mathrm{SP},2}$. There is no direct period-2 learning externality in this two-period lagged-learning model. The planner's extra term appears in the first-period marginal condition because period-1 data raises both firms' period-2 productivity. Period-2 underinvestment is inherited from the period-1 data-generation wedge.
\end{proof}

\subsubsection{Proof of Proposition~\ref{prop:reversal}}
\label{app:proof-prop-reversal}
\begin{proof}
At any common $x$, $F_{\mathrm{NP}}(x)-F_{\mathrm{NF}}(x)=\delta D(x)$ by direct subtraction of \eqref{eq:F-NF}--\eqref{eq:F-NP}. Since $F_{\mathrm{NF}}(K^*_{\mathrm{NF},1})=0$, $F_{\mathrm{NP}}(K^*_{\mathrm{NF},1})=\delta D(K^*_{\mathrm{NF},1})$. The function $F_{\mathrm{NP}}$ is strictly decreasing with unique root $K^*_{\mathrm{NP},1}$, so the root lies strictly to the right of $K^*_{\mathrm{NF},1}$ if and only if $F_{\mathrm{NP}}(K^*_{\mathrm{NF},1})>0$, which is equivalent to $D(K^*_{\mathrm{NF},1})>0$ because $\delta>0$. The equality case is identical.
\end{proof}

\subsubsection{Proof of Proposition~\ref{prop:cs}}
\label{app:proof-prop-cs}
\begin{proof}
\emph{Parts (a) and (c).} By \eqref{eq:welfare}, $\Delta_{\mathrm{pool}}=2[V_{\mathrm{SP}}(x_{\mathrm{NP}})-V_{\mathrm{NF}}(x_{\mathrm{NF}})]$. Lemma~\ref{lem:ladder} gives the inequality, weak in general and strict whenever $\eta$ is strictly increasing. Similarly $\Delta_{\mathrm{coord}}=2[V_{\mathrm{SP}}(x_{\mathrm{SP}})-V_{\mathrm{SP}}(x_{\mathrm{NP}})]\geq 0$ since $x_{\mathrm{SP}}$ maximizes $V_{\mathrm{SP}}$ and strict whenever $x_{\mathrm{NP}}\neq x_{\mathrm{SP}}$. By Proposition~\ref{prop:underinvestment}, $x_{\mathrm{NP}}<x_{\mathrm{SP}}$ whenever $\eta'>0$ on the relevant domain. \smallskip \\
\emph{Part (b).} With $\eta(N)\equiv\eta_{\min}$, $\eta'\equiv 0$. The continuation value reduces to $\phi(K,N)\equiv rK\eta_{\min} + r^2\eta_{\min}^2/(2c_2)$, which is independent of $N$. So, $V_{\mathrm{NF}}, V_{\mathrm{SP}}, V_i^{\mathrm{NP}}(\cdot,x_j)$ coincide as functions of $x$ (and of $x_i$, respectively) and the three FOCs collapse to $r\eta_{\min}-c_1 x + \delta r\eta_{\min}=0$, with common root $x_{\mathrm{NF}}=x_{\mathrm{NP}}=x_{\mathrm{SP}}=(1+\delta)r\eta_{\min}/c_1$. Therefore $W^*_{\mathrm{NF}}=W^*_{\mathrm{NP}}=W^*_{\mathrm{SP}}$ and both wedges vanish.
\end{proof}

\subsubsection{Proof of Proposition~\ref{prop:threshold}}
\label{app:proof-prop-threshold}
\begin{proof}
\emph{Part (a).} As $c_2 \to \infty$, the symmetric Cournot solution in Lemma~\ref{lem:p2-cournot} satisfies $K_2^*(\eta, K) \to K$, since
\[
K_2^*(\eta, K) - K \;=\; \frac{a\eta - 3 b \eta^2 K}{3 b \eta^2 + c_2} \;\to\; 0,
\]
and the expansion-cost term satisfies $(c_2/2)(K_2^* - K)^2 = (c_2/2)(a\eta - 3 b \eta^2 K)^2/(3 b \eta^2 + c_2)^2 \to 0$. So, in this limit
\[
v_2(\eta, K) \;\longrightarrow\; K \eta \bigl(a - 2 b K \eta\bigr).
\]
Substituting, we have
\[
v_2(\eta^P, K) - v_2(\eta^F, K) \;=\; K(\eta^P - \eta^F)\,\bigl[\,a - 2 b K (\eta^P + \eta^F)\,\bigr]. 
\]
Strict monotonicity of $\eta$ gives $\eta^P > \eta^F$, so the sign of the difference in \eqref{eq:stage0-reduced} equals the sign of $a - 2 b K(\eta^P + \eta^F)$, which yields the threshold \eqref{eq:b-star}. \smallskip \\
\emph{Part (b).} The map $b \mapsto K_2^*(\eta, K) = (a\eta + c_2 K)/(3 b \eta^2 + c_2)$ is continuous in $b$ on the stated interval, therefore so is $b \mapsto v_2(\eta, K)$ for each fixed $\eta$. At $b = 0$, the inverse demand is constant at $a$ and Lemma~\ref{lem:p2-cournot} gives $K_2^*(\eta, K) = K + a\eta/c_2$. Substituting into \eqref{eq:v2-def},
\[
v_2(\eta, K)\big|_{b=0} \;=\; a K \eta + \frac{a^2 \eta^2}{2 c_2},
\]
which is strictly increasing in $\eta$. Combined with $\eta^P > \eta^F$, this gives $U^P(0^+) - U^F(0^+) > 0$, and by continuity the inequality persists in a right-neighborhood of $b = 0$. Moreover, differentiating $v_2$ in \eqref{eq:v2-def} with respect to $b$ and evaluating at $b = 0$, the terms in $\partial K_2^*/\partial b$ cancel and $\partial v_2/\partial b|_{b=0} = -2\,\eta^2(K + a\eta/c_2)^2$. Since $\eta \mapsto \eta K + a\eta^2/c_2$ is strictly positive and strictly increasing on $\eta > 0$, its square is too, so with $\eta^P > \eta^F$ the difference $\partial_b(U^P - U^F)|_{b=0}$ is strictly negative. The sign of $U^P - U^F$ away from $b = 0$ is governed by the monotonicity of $v_2(\cdot, K)$ in $\eta$. Writing the period-2 payoff before symmetry as $\pi_2(K_i, K_j; \eta) = [a - b\eta(K_i + K_j)] K_i \eta - (c_2/2)(K_i - K)^2$ and using the firm's own period-2 first-order condition gives the envelope decomposition
\[
\frac{\partial v_2}{\partial \eta}
\;=\; \underbrace{K_2^*\bigl(a - 4 b \eta K_2^*\bigr)}_{\text{direct effect}}
\;-\; \underbrace{b\,\eta^2 K_2^*\,\frac{\partial K_2^*}{\partial \eta}}_{\text{strategic effect}},
\]
where the direct effect is the gain from a firm's own higher productivity and the strategic effect captures the price-mediated impact of the rival's equilibrium capacity response, absent under exogenous pricing. Neither term is signed unconditionally. The direct effect is positive for sufficiently small $b$ but can become negative for high-productivity states such as $\eta=\eta^P$, while the strategic effect has the opposite sign of $\partial K_2^*/\partial\eta$. This ambiguity precludes a closed-form threshold for finite $c_2$.
\end{proof}

\subsubsection{Proof of Proposition~\ref{prop:wedge}}
\label{app:proof-prop-wedge}
\begin{proof}
The period-1 component of $\partial\pi_i^R/\partial K_i$, evaluated at $K_i = K_j = K$, equals $a\eta_{\min} - (3 b \eta_{\min}^2 + c_1)K$ and is common to both regimes. The period-2 component is $\delta\, g^R(K)$, with $g^F$ and $g^P$ obtained by differentiating the period-2 term of \eqref{eq:wedge-profit} and imposing symmetry. This gives \eqref{eq:wedge-foc}. Subtracting,
\[
F^P(K) - F^F(K) = \delta\bigl(g^P(K) - g^F(K)\bigr) = \delta\Delta(K),
\]
and collecting the terms of $g^P - g^F$ by their dependence on $a$ and on $b$ yields $\Delta = \Delta_0 + \Delta_{\mathrm{comp}}$. Under Assumption~\ref{ass:stage1-reg}, $F^P$ is strictly decreasing with unique root $K^P_\star$, and $F^F(K^F_\star) = 0$. So, $F^P(K^F_\star) = F^P(K^F_\star) - F^F(K^F_\star) = \delta\,\Delta(K^F_\star)$, and strict monotonicity of $F^P$ gives the stated three-way comparison since $\delta > 0$. \smallskip \\
For (i), every term of $\Delta_{\mathrm{comp}}$ is proportional to $b$, so $\Delta(K)|_{b=0} = \Delta_0(K) = a[\eta(2K) - \eta(K) + K(\eta'(2K) - \eta'(K))]$. This is the $c_2 \to \infty$ limit of $D$ in Proposition~\ref{prop:reversal} of the baseline after the substitution $r \mapsto a$, and that $D$ admits both signs is established there. For (ii), $\eta(2K) > \eta(K) > 0$ by strict monotonicity and positivity of $\eta$, with $b, K > 0$. For (iii), the stated inequality makes the second term of $\Delta_{\mathrm{comp}}$ non-positive while the first is strictly negative by (ii), so $\Delta_{\mathrm{comp}}(K) < 0$.
\end{proof}

\subsubsection{Proof of Proposition~\ref{prop:general-threshold}}
\label{app:proof-prop-general-threshold}
\begin{proof}
Immediate from \eqref{eq:general-sustainability} and $\delta > 0$.
\end{proof}

\end{document}